\documentclass[letterpaper, 10 pt, conference]{ieeeconf}  

\IEEEoverridecommandlockouts                              

\usepackage{cite}
\usepackage[pdftex]{graphicx}
\usepackage{caption}
\usepackage[font=scriptsize]{caption}
\usepackage{amsmath}
\newtheorem{theorem}{Theorem}

\usepackage{txfonts}

\title{\LARGE \bf
Performance Analysis of Event-Triggered Consensus Control for Multi-agent Systems under Cyber-Physical Attacks
}

\author{ Farzaneh Tatari$^{1}$,  Aquib Mustafa$^{1}$, Majid Mazouchi$^{1}$, Hamidreza Modares$^{1}$, Christos G. Panayiotou$^{2}$ \\ and  Marios M. Polycarpou$^{2}$
\thanks{}
\thanks{F.~Tatari$^{1}$, A.~Mustafa, M.~Mazouchi, and H.~Modares  are with the Department of Mechanical Engineering, Michigan State University, East Lansing, MI, 48824, USA
        {\tt\small tatarifa@msu.edu; mustaf15@msu.edu; mazouchi@msu.edu; modaresh@msu.edu}}%
\thanks{$^{2}$M.~M. Polycarpou and C.~G. Panayiotou are with the KIOS Research and Innovation Center of Excellence and the Department of Electrical and Computer Engineering, University of Cyprus, Nicosia, 1678, Cyprus
        {\tt\small christosp@ucy.ac.cy; mpolycar\allowbreak @ucy.ac.cy}}%
}

\begin{document}

\maketitle
\thispagestyle{empty}
\pagestyle{empty}

\begin{abstract}

This work presents a rigorous analysis of the adverse effects of cyber-physical attacks on the performance of multi-agent consensus with event-triggered control protocols. It is shown how a strategic malicious attack on sensors and actuators can deceive the triggering condition of both state-based event-triggered mechanism and combinational state-based event-triggered mechanism, which are commonplace and widely used in the literature. More precisely, it is first shown that a deception attack in the case of combinational state-based event-triggered mechanism can result in a non-triggering misbehavior, in the sense that the compromised agent does not trigger any event and consequently results in partial feedback disconnectivity by preventing information from reaching the local neighbors of the compromised agent. This indicates that the combinational state-based event-triggered mechanism can be leveraged by the attacker to harm the network connectivity by rendering the recent data unavailable to agents. It is then shown that the deception attack in the case of state-based event-triggered mechanism can result in a continuous-triggering misbehavior in the sense that  the event-triggered mechanism continuously generates triggering events resulting in undesirable phenomenon of Zeno behavior. Finally, numerical simulations are presented to illustrate the theoretical findings.

\end{abstract}

\section{INTRODUCTION}

The success of next-generation connected autonomous vehicles can profoundly impact the transport sector globally by improving efficiency, preventing traffic congestion, and increasing road safety. In the past years, toward this goal, a rich body of work has considered designing distributed multi-agent systems (MASs) that leverage communication between agent (vehicles) to further reduce congestion. Generally, in distributed control of MASs, a set of agents communicate with each other over a communication network to reach a coordinated group behavior. Specifically, in the consensus control problem, the coordinated group behavior is specified as achieving agreement between agents on some quantity of interest. 




Traditional approaches for implementation of distributed control protocols of MASs require continuous exchange of information among agents, which demands a large communication bandwidth and can result in congestion for systems with resource-limited embedded micro-processors. To alleviate this issue, several real-time scheduling methods, called event-triggered control approaches, have been introduced for reducing the communication burden while keeping the performance at an acceptable level \cite{Tabuada2007,Dimos2012TAC,Yi2019Automatica}. Typically, communication and control updates are executed when the ratio of a certain error norm exceeds a predefined threshold. A rich body of literature has been developed on event-triggered control approaches \cite{zhu2014auto,cheng2017TAC,Hu2017cyber,VAMVOUDAKIS2018412}. In these approaches  it is assumed that the communication network and agents are reliable and not compromised. This assumption, nevertheless, can be violated in the presence of cyber-physical attacks.

The design of secure event-triggered controllers has been recently studied in \cite{Feng2017ACC,Ding2017cyber,Shoukry2016TAC}. In \cite{Feng2017ACC}, a secure average consensus problem is considered for linear MASs subject to denial of service (DoS) attacks. It is shown in \cite{Feng2017ACC} that resilience against DoS attacks is accomplished as long as the frequency and period of DoS attacks satisfy certain conditions. In \cite{Ding2017cyber}, an event-triggered  control approach is introduced for the consensus of MASs with lossy sensors under cyber-physical attacks. The cyber-physical attacks considered in \cite{Ding2017cyber} are deception attacks which are represented by bounded disturbances occurring randomly. In \cite{Shoukry2016TAC}, the authors develop an event-based algorithm to mitigate attacks on the estimator by reconstructing both the state as well as the sparse attack signal for single-agent discrete-time linear systems. No control protocol is assumed in \cite{Shoukry2016TAC}.


Despite significant progress in designing resilient distributed control systems \cite{Dibaji2018IFAC,Sandberg2015IEEE,Pasqualetti2015,Vamvoudakis2018,Teixeira2015Automatica,Etesami2019,amini2018resilient,Tsiakkas2018CDC,du2021resilient}, there is no rigorous analysis on how a stealthy attack can leverage the event-triggered mechanism in distributed control protocols to undermine the system performance. The science of modeling and analysis of adverse effects of attacks is a key step in securing the system by identifying possible vulnerabilities of event-triggered mechanisms, which helps to design enhanced defense mechanisms against them. Towards this aim, in this paper, the effect of cyber-physical attacks on event-triggered based distributed control methods is rigorously studied. Two general event mechanisms, namely, state-based and combinational state-based event mechanisms, are reviewed, and, then, their vulnerabilities to attacks are investigated. First, we show that the deception attack on a combinational state-based event-triggered mechanism (CS-ETM) may result in a non-triggering misbehavior. In this case, even if the performance of the compromised agent is far from the desired, the event triggering mechanism is fooled by the attacker and consequently does not trigger any event, virtually disconnecting all its outcoming links. Hence, this deception attack acts like a denial-of-service (DoS) attack, which can greatly harm the speed of the information flow within agents and the network connectivity. Moreover, we show that the deception attack on a state-based event-triggered mechanism (S-ETM) can lead to a continuous-triggering misbehavior, where the event-triggered mechanism continuously generates events, resulting in Zeno behavior \cite{Ames44CDC}. Zeno behavior is an extremely undesirable phenomenon in hybrid systems since the system is forced to sample excessively fast, and as result the execution instants get too close to each other, causing, a countless number of discrete transitions to take place in a finite time interval \cite{Ames44CDC}. Moreover, such deception attacks on S-ETM may exhaust the agents computational or communication resources in the cyber layer.

\medskip

\noindent
\textbf{Notations:} The following notations are used throughout this paper. ${\mathbb{R}^n}$ and ${\mathbb{R}^{n \times m}}$ represent, respectively, the $n$-dimensional real vector space, and the $n \times m$ real matrix space.  Let ${1_n}$ be the column vector with all entries equal to one. ${\mathcal{I}_n}$ represents the  $n \times n$ identity matrix. $diag\left( {{d_1},...,{d_n}} \right)$ represents a block-diagonal matrix with matrices ${d_1},...,{d_n}$ on its diagonal. The symbol $ \otimes $ represents the Kronecker product, while ${\left\| . \right\|}$ denotes the Euclidean norm. Given a matrix $E \in {{\mathbb{R}}^{m \times n}}$, ${(E)^T} \in {{\mathbb{R}}^{n \times m}}$ denotes its transpose. ${\mathcal U}(a,{\rm{ }}b)$  with $a<b$ denotes an uniform distribution between the interval $a$ and $b$. Finally, $X \sim {\cal U} (a,{\rm{ }}b)$ denotes that $X$ is distributed uniformly with a probability density function of ${f_X}\left( x \right) = {1 \mathord{\left/
 {\vphantom {1 {(b - a)}}} \right.
 \kern-\nulldelimiterspace} {(b - a)}}:{\mkern 1mu} a < x < b$.

\section{PRELIMINARY}
In this section, some background on the graph theory and the problem formulation are provided.

\subsection{Graph Theory} 
A graph ${\cal G}$ with $N$ nodes consists of a pair $\left( {{\cal V},{\cal E}} \right)$, in which ${\cal V}{\rm{  =   }}\{ {v_1}, \cdots ,{v_N}\} $ is the set of nodes and  ${\cal E} \subseteq {\cal V} \times {\cal V}$ is the set of edges. The adjacency matrix is defined as $\mathcal{A} = [ {{a_{ij}}}]$, with  ${a_{ij}} = 1$ if $({v_j},{v_i}) \in {\cal E}$, and ${a_{ij}} = 0$, otherwise.  A graph is undirected if ${a_{ij}} = {a_{ji}}$. The nodes ${\nu _i}$ and  ${\nu _j}$ are adjacent if there exists an edge between them. The set of neighbors of node ${\nu _i}$  is denoted by $N_i^I = \{ {{\nu _j} \in {\cal V}:( {{\nu _j},{\nu _i}} ) \in {\cal E},j \ne i} \}$. A path from node ${\nu _i}$ to node ${\nu _j}$ is a sequence of distinct nodes starting from node ${\nu _i}$ and ending with node ${\nu _j}$ while consecutive nodes are adjacent. An undirected graph ${\cal G}$ is said to be connected if there exists a path between every pair of nodes. A node ${\nu _i}$ is said to be reachable from a node ${\nu _j}$ if there exists a path from ${\nu _j}$ to ${\nu _i}$. In a connected graph, all nodes are reachable from one another.  The in-degree matrix of the graph ${\cal G}$ is defined as $D = diag({d_i})$, where  ${d_i} = \sum\nolimits_{j \in N_i^I} {{a_{ij}}} $ is the weighted in-degree of node ${\nu _i}$.  The graph Laplacian matrix of ${\cal G}$ is defined as ${\cal L} = D - \mathcal{A}$. 

\subsection{Problem Formulation} 
We consider a MAS composed of $N$ agents having identical dynamics given by
\begin{equation}
{\dot x_i}(t) =  A{x_i}(t) + B{u_i}(t), \label{eq:1}
\end{equation}
where ${x_i}(t) \in {\mathbb{R}^n}$ and ${u_i}(t) \in {\mathbb{R}^m}$ denote the state and control input, respectively, and $A$ and $B$ are the drift and input dynamics, respectively.

\noindent
\quad \textbf{Assumption 1.} The communication graph ${\cal G}$ is undirected and connected with no self-connections.

\noindent
\quad \textbf{Assumption 2}. $(A,B)$ in (\ref{eq:1}) is stabilizable.

\noindent
\quad \textbf{Problem 1.} The objective of distributed control is to design local controllers for each agent, i.e., ${u_i(t)}$ in \eqref{eq:1}, so that agents reach consensus. That is,
\begin{equation}\label{eq:2}
\mathop {\lim }\limits_{t \to \infty } {\text{ }}( {{x_i}(t) - {x_j}(t)} ) \to 0\begin{array}{*{20}{c}}
  {}&{\forall i,j = 1, \ldots ,N}.
\end{array}
\end{equation}

Several event-triggered control protocols are designed in the literature to solve Problem 1, two of which are reviewed in the next section. However, all these results assumed a reliable network and ignored cyber-physical attacks. 

\smallskip

\section{EVENT-TRIGGER BASED CONSENSUS FOR DISTRIBUTED MULTI-AGENT SYSTEMS}
There are two main approaches for event-triggered control protocols for MASs: the event-triggered mechanism based on the error in the agent's local neighborhood tracking error, which we call combinational state-based event-triggered mechanism (CS-ETM), and the event-triggered mechanism based on the error in the agent's state, which we call state-based event-triggered mechanism (S-ETM). For completeness, in this section, both CS-ETM and S-ETM are reviewed.

\subsection{Combinational state-based event mechanism} 
In the CS-ETM approach, the measurement error is obtained by
\begin{equation}\label{eq:7}
{\bar e_i}(t) = {q_i}(t)-{q_i}(t_k^i),
\end{equation}
where
\begin{equation}\label{eq:7*}
 {q_i}(t) = \sum\limits_{j \in {N_i^I}} {({x_j}(t) - {x_i}(t))},
\end{equation}
\begin{equation}\label{eq:7**}
  {q_i}(t_k^i) = \sum\limits_{j \in {N_i^I}} {({x_j}(t_{k}^i) - {x_i}(t_{k}^i))},
 \end{equation}
 where ${q_i}(t)$ and ${q_i}(t_k^i)$ denote the local neighborhood tracking error in time instants $t$ and $t_k^i$, respectively. Then, using the measurement error \eqref{eq:7}, the triggering condition is obtained as \cite{Fang2013,Wenfeng2016}
\begin{equation}\label{eq:8}
\left\| {\bar {e_i}}(t) \right\| \ge {\eta _i}\left\| {{q_i}}(t) \right\|,
\end{equation}
where $0 < {\eta _i} < 1$, and the control protocol becomes
\begin{equation}\label{eq:9}
{u_i}(t) = K{q_i}(t_k^i),\begin{array}{*{20}{c}}
  {}&{} 
\end{array}t \in [ {t_k^i,t_{k + 1}^i} ).
\end{equation}
\smallskip
where $K \in {\mathbb{R}^{m \times n}}$ is the feedback control gain matrix, and $K$ is designed such that $A-{\lambda _i}BK$, $\forall i = 2,{\rm{ }}.{\rm{ }}.{\rm{ }}.{\rm{ }},N$ become Hurwitz where $\lambda _i$, $\forall i = 2,{\rm{ }}.{\rm{ }}.{\rm{ }}.{\rm{ }},N$ are the nonzero eigenvalues of the
graph Laplacian matrix $\cal L$.

\subsection{State-based event-triggered mechanism} 
 In the S-ETM approach, the measurement error depends on the agent's state, given by
\begin{equation}\label{eq:4}
{e_i}(t) = {x_i}(t_k^i) - {x_i}(t),
\end{equation}
where $t_k^i$ denotes the ${k}$-th triggering event of agent $i$ and ${x_i}(t_k^i)$ is the state of agent $i$ at the triggering event $t_k^i$. Now, let $A=0$, $B=1$, $m=1$, and $n=1$. Then, the triggering condition is given by \cite{Dimos2012TAC}
\begin{equation}\label{eq:5}
(e_i(t))^2 \ge {\eta _i}{( {\sum\limits_{j \in {N_i^I}} { ({{x_i}(t) - {x_j}(t)} )} })^2}),
\end{equation}
where $0 < {\eta _i} < 1$. The condition \eqref{eq:5} is only based on the relative information of $i$-th agent's neighbors. In this case, the control protocol for each agent is 
\begin{equation}\label{eq:6}
{u_i}(t) =  -K \sum\limits_{j \in {N_i^I}} {( {{x_i}(t_k^i) - {x_j}(t_{k'}^j)} )}, \,\,\,\,\,\,\,\,t \in [t_k^i,t_{k + 1}^i),
\end{equation}
where $K \in {\mathbb{R}^{m \times n}}$ is the feedback control gain matrix, and $K$ is designed such that $A-{\lambda _i}BK$, $\forall i = 2,{\rm{ }}.{\rm{ }}.{\rm{ }}.{\rm{ }},N$ become Hurwitz where $\lambda _i$, $\forall i = 2,{\rm{ }}.{\rm{ }}.{\rm{ }}.{\rm{ }},N$ are the nonzero eigenvalues of the
graph Laplacian matrix $\cal L$. Moreover, $k^{\prime} := \arg \min _{l \in \mathbb{N}: t \geq t_{l}^{j}}\left\{t-t_{l}^{j}\right\}$, $t_{k'}^j$ denotes the $k'$-th triggering event of agent $j$ and ${x_j}(t_{k'}^j)$ is the state of agent $j$ at the triggering event $t_{k'}^j$.
\smallskip

\noindent
\section{ATTACK ANALYSIS}
In this section, the attack analysis for event-triggered based consensus of MASs is investigated.
It is shown in the following that an attacker that has resources to get access to some knowledge about the agents' dynamics and graph topology hereafter called strategic attacker, can cause non-triggering misbehavior in CS-ETM and continuous triggering misbehavior and consequently the Zeno behavior in S-ETM.  

The deception attack on sensors and actuators of compromised agents is respectively modeled as 

\begin{equation}\label{eq:100}
x_i^c(t) = {x_i}(t) + {\beta _i}x_i^a(t),
\end{equation}
with
\begin{equation}\label{eq:101}
{\beta _i} = \left\{ {\begin{array}{*{20}{c}}
1&{Agent\,i\,is\,under\,sensor\,attack}\\
0&{Otherwise,}
\end{array}} \right.
\end{equation}
while
\begin{equation}\label{eq:102}
u_i^c(t) = {u_i}(t) + f_i(t),
\end{equation}
\begin{equation}\label{eq:102*}
f_i(t)={\alpha _i}u_i^a(t),
\end{equation}
\begin{equation}\label{eq:103}
{\alpha _i} = \left\{ {\begin{array}{*{20}{c}}
1&{Agent\,i\,is\,under\,actuator\,attack}\\
0&{Otherwise,}
\end{array}} \right.
\end{equation}
where ${x_i} \in {\mathbb{R}^n}$ is the normal state, $x_i^a \in {\mathbb{R}^n}$ denotes the attack signal inserted into the state of agent $i$, $x_i^c \in {\mathbb{R}^n}$ is the manipulated measurement, ${u_i} \in {\mathbb{R}^m}$ is the nominal control protocol, $u_i^a \in {\mathbb{R}^m}$ denotes the attack signal inserted into the actuators of agent $i$, and $u_i^c \in {\mathbb{R}^m}$ is the compromised control protocol applied to agent $i$.
\medskip

Fig. \ref{fig:Fig1} shows the overall structure of an event-triggered control protocol for the MAS under attack.

\begin{figure}[!ht]
\begin{center}
\includegraphics[width=3.3in,height=2in]{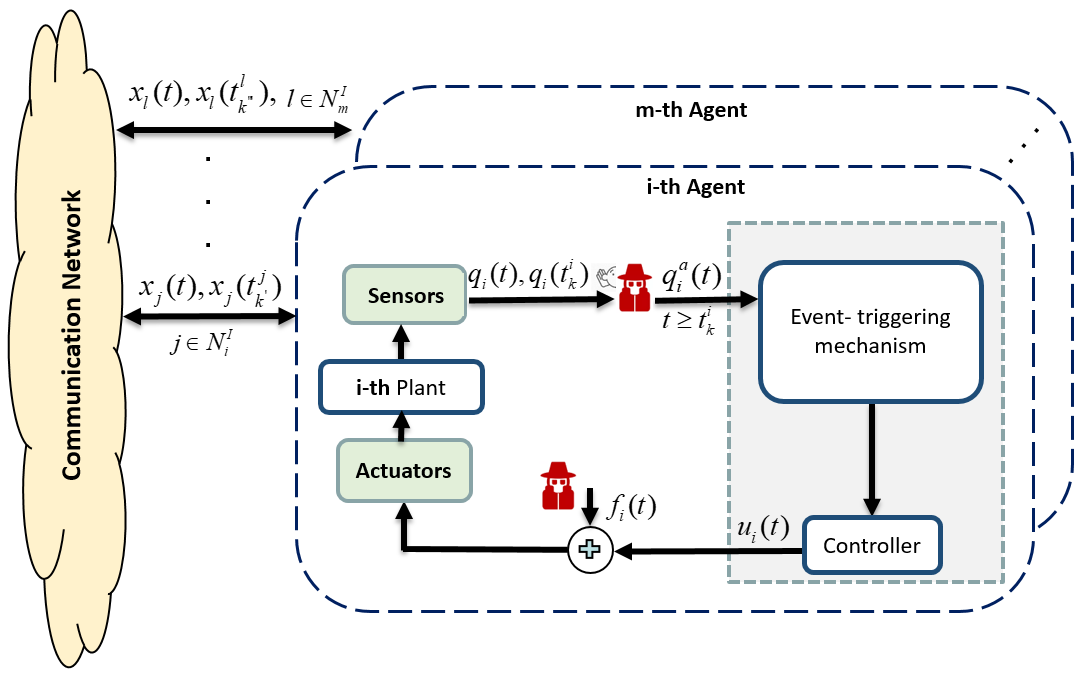}
\vspace{-5pt}
\caption{The overall structure of a strategic attack on event-triggered control protocols.}\label{fig:Fig1}
\captionsetup{justification=centering}
\vspace{0pt}
\end{center}
\end{figure}
\vspace{-5pt}
\medskip

\subsection{Strategic sensor attack effect on CS-ETM}


Using \eqref{eq:1}, \eqref{eq:7*}, \eqref{eq:100} and \eqref{eq:102}, the dynamics of the compromised agent $i$ and the compromised local neighborhood tracking error at time instant $t$, denoted by $q_i^c(t)$, become
\begin{equation}\label{eq:10}
{\dot x_i}(t) = A{x_i}(t) + B{u_i}(t) + B{f_i}(t),
\end{equation}
and
\begin{equation}\label{eq:1000}
q_i^c(t) = {q_i}(t) + {\beta _i}q_i^m(t),
\end{equation}
where $q_i^m(t) =  - {d_i}x_i^a(t)$ is the local neighborhood tracking error of agent $i$ and $d_i$ is the weighted in-degree of agent $i$.

The aim of the attacker is to harm the compromised agent or even the entire network by deceiving the triggering mechanism and consequently to degrade the performance of the MAS. To deceive an agent into exhibiting a non-triggering misbehavior, a strategic attacker can tamper the sensor reading, i.e., ${q_i}(t)$, by scheming a replay attack.
The following theorem shows that a strategic attacker can launch an attack on the sensors of an agent and change the consensus value without causing the violation of the triggering condition \eqref{eq:8}. Furthermore, this malicious attack strategy can cause non-triggering misbehavior which can make the original graph ${\cal G}$ clustered into several subgraphs and therefore harm the communication graph connectivity. To show this, we first need the following definition. 

\noindent
\quad\textbf{Definition 1. (Vertex cut).} A vertex cut of ${\cal G}$ is a set of nodes  $\mathcal V(\Xi ) \subseteq \mathcal V\left(\cal G \right)$ such that removing them from the graph ${\cal G}$, i.e., ${\cal G} \setminus \Xi $, results in a disconnected graph clusters.

\begin{theorem} \label{theor5}
Consider the MAS \eqref{eq:10} under the control protocol \eqref{eq:9}. Let a set of agents $\mathcal V(\Xi ) \subseteq \mathcal V\left(\cal G \right)$ be the vertex cut of the graph ${\cal G}$ and all of its members be under a strategic replay attack, for $t \ge t_k^i$, given by 
\begin{align} \label{eq:300}
q_i^c(t) = q_i^E(t_k^i) + {\theta _i^a} 1_n,
\end{align}
where  $q_i^E(t_k^i) = {q_i}(t_k^i)$ is the eavesdropped communicated information ${q_i}(t)$ at the triggering time instant $t_k^i$,\\
and ${\theta _i^a}  \sim  {\cal U}(a_{i},b_{i})$
 with
 \begin{align} \label{eq:350}
 a_{i}={(1 + {\eta _i})^{ - 1}}\left\| {{q_i}(t_k^i)} \right\| - \left\| {q_i^E(t_k^i)} \right\|,
  \end{align}
  \begin{align} \label{eq:351}
  b_{i}={(1 - {\eta _i})^{ - 1}}\left\| {{q_i}(t_k^i)} \right\| - \left\| {q_i^E(t_k^i)} \right\|,
  \end{align}
  is a uniformly distributed random number, where $\eta_i$ is the event threshold defined in \eqref{eq:8}. Then, the triggering condition \eqref{eq:8} can never be violated for agent $i$, $\forall i \in \mathcal V(\Xi )$ and consequently, it shows non-triggering misbehavior which makes the original graph ${\cal G}$ clustered into several subgraphs.
\end{theorem}

\begin{proof}
The proof is carried out in two steps. In the first step, we show that the condition \eqref{eq:8} can never be violated for agent $i$, $\forall i \in \mathcal V(\Xi )$, if 
\begin{equation}\label{eq:14***}
{(1 + {\eta _i})^{ - 1}}\left\| {{q_i}(t_k^i)} \right\| < \left\| {{q_i}(t)} \right\| < {(1 - {\eta _i})^{ - 1}}\left\| {{q_i}(t_k^i)} \right\|
\end{equation}
holds. In the second step, we use \eqref{eq:14***} to conclude the proof. To this aim, using \eqref{eq:7} and some manipulation, \eqref{eq:8} can be rewritten as
\begin{equation}\label{eq:14*}
(1 - {\eta _i}^2){\left\| {{q_i}(t)} \right\|^2} - 2\left\| {{q_i}(t)} \right\|\left\| {{q_i}(t_k^i)} \right\| + {\left\| {{q_i}(t_k^i)} \right\|^2} \le 0.
\end{equation}
Note that ${q_i}(t_k^i)$ is constant for $t \in \left[ {t_k^i,t_{k + 1}^i} \right)$. Therefore, it follows from the quadratic formula that
\begin{align} \label{eq:14**}
(1 - {\eta _i}^2){\left\| {{q_i}} \right\|^2} - 2\left\| {{q_i}(t_k^i)} \right\|\left\| {{q_i}} \right\| + {\left\| {{q_i}(t_k^i)} \right\|^2} =&  \nonumber \\
(1 - {\eta _i}^2)(\left\| {{q_i}} \right\| - \frac{{\left\| {{q_i}(t_k^i)} \right\|}}{{(1 - {\eta _i})}})(\left\| {{q_i}} \right\| -& \frac{{\left\| {{q_i}(t_k^i)} \right\|}}{{(1 + {\eta _i})}}).
\end{align}
Completing the square, one can see that under the strategic replay attack given by \eqref{eq:300}, the event-triggered condition in \eqref{eq:8}, $\forall i \in \mathcal V(\Xi )$, can never be violated based on \eqref{eq:14***}, and thus one has
\begin{align} \label{eq:301}
\left\| {{e_i}}(t) \right\|= \left\| q_i^c(t)-q_i(t_{k}^{i}) \right\|  \le {\eta _i}(t)\left\| q_i^c(t) \right\|, t \ge t_{k}^{i}, \forall i \in \mathcal V(\Xi ).
\end{align}
Then, these agents do not transmit their information to their neighbors and they act as sink agents. Since the set of agents $\mathcal V(\Xi )$ is a vertex cut, the non-triggering misbehavior of agents in $\mathcal V(\Xi )$   prevents the exchange of information from happening between agents in some segments of the graph $\cal G$, and therefore clusters the original graph $\cal G$ into some subgraphs. This completes the proof.
\end{proof}





\subsection{Strategic actuator attack effect on S-ETM}
In this subsection, it is shown that for the S-ETM, a strategic malicious attack on the actuator can cause the entire network to exhibit a continuous-triggering misbehavior. 

\begin{theorem} \label{Theor2}
Consider the MAS \eqref{eq:10} with the control protocol \eqref{eq:6}, subject to a strategic malicious attack on its actuator; i.e., ${\beta _i} = 0$ and ${\alpha _i} = 1$. Then, injecting an attack signal with sufficiently large magnitude into the agent $i$ can cause the network to exhibit a continuous-triggering misbehavior which results in Zeno behavior.
\end{theorem}

\begin{proof}
Taking the Laplace transform and using \eqref{eq:4}, \eqref{eq:6} and \eqref{eq:102}, the global control signal becomes
\begin{equation}\label{eq:420}
u(s) =  - ({\cal L} \otimes {{K}})(e(s) + x(s)) + f(s),
\end{equation}
where $e(s) = {\left[ {{e_1}(s), \ldots ,{e_N}(s)} \right]^T}$, $x(s) = {\left[ {{x_1}(s), \ldots ,{x_N}(s)} \right]^T}$, $u(s) = {\left[ {{u_1}(s), \ldots ,{u_N}(s)} \right]^T}$ and $f(s) = {\left[ {{f_1}(s), \ldots ,{f_N}(s)} \right]^T}$, denote the global vector of measurement error, agent's state, signal control and attack signal, respectively. 
Define now the transfer function of the system \eqref{eq:1} from ${x_i}\left( s \right)$ to ${u_i}\left( s \right)$ as
\begin{equation}\label{eq:421}
G(s) = \frac{{{x_i}(s)}}{{{u_i}(s)}} = {(s{\cal I}_n - A)^{ - 1}}B.
\end{equation}
Utilizing \eqref{eq:420} and \eqref{eq:421}, the overall state can be written as
\begin{equation}\label{eq:422}
x(s) = ({{\cal I}_N} \otimes G(s))[ - ({\cal L} \otimes {{K}})(x(s)) + (f(s) - ({\cal L} \otimes {{\cal I}_n})e(s))].
\end{equation}

Let $T = [{v_{ij}}] \in {\mathbb{R}^{N \times N}}$ be a matrix of right eigenvectors of  $\cal L$ where its first column is the right eigenvector corresponding to $0$, then ${T^{ - 1}} = [{w_{ij}}] \in {\mathbb{R}^{N \times N}}$ is the matrix of its left eigenvectors and one has
\begin{align} \label{eq:423}
{\cal L} = T\Lambda {T^{ - 1}},
\end{align}
with $\Lambda $ being the Jordan canonical form of the graph Laplacian matrix $\cal L$. Utilizing \eqref{eq:423}, \eqref{eq:422} can be rewritten as follows
\begin{align}
&(T \otimes {I_n})\left[ {{{\cal I}_{Nn}} + (\Lambda  \otimes G(s)K)} \right]({T^{ - 1}} \otimes {{\cal I}_n})x(s) = \nonumber \\ 
& \qquad \qquad \qquad \qquad ({{\cal I}_N} \otimes G(s))[f(s) - ({\cal L} \otimes {{K}})e(s)]. \label{eq:424}
\end{align}
Multiplying both sides of \eqref{eq:424} by $({T^{ - 1}} \otimes {{\cal I}_n})$ and using the state transformation as $\hat x(s) = ({T^{ - 1}} \otimes {{\cal I}_n})x(s)$, one has
\begin{align}
&\left[ {{{\cal I}_{Nn}} + (\Lambda  \otimes G(s)K)} \right]\hat x(s) = \nonumber \\ 
&({{\cal I}_N} \otimes G(s))[({T^{ - 1}} \otimes {{\cal I}_n})f(s) - (\Lambda  \otimes {{K}})({T^{ - 1}} \otimes {{\cal I}_n})e(s)], \label{eq:1425}
\end{align}
which can be written as
\begin{align} \label{eq:425}
&\hat x(s) = {\left[ {{{\cal I}_{Nn}} + (\Lambda  \otimes G(s)K)} \right]^{ - 1}} \times \nonumber \\
& \qquad [({T^{ - 1}} \otimes G(s))f(s) - (\Lambda  \otimes G(s)K)({T^{ - 1}} \otimes {{\cal I}_n})e(s)].
\end{align}
Note that \eqref{eq:425} is a block diagonal system and the size of each block is identical to the size of the Jordan block associated with an eigenvalue ${\lambda _i}$ of the graph Laplacian matrix $\cal L$. Now, without loss of generality, assume that all Jordan blocks in $\Lambda$ are simple. Then, it follows from \eqref{eq:425} that
\begin{align}
{\hat x_i}(s) =& {\left[ {{{\cal I}_n} + {\lambda _i}G(s)K} \right]^{ - 1}}G(s) \times \nonumber \\
&\qquad \qquad [\sum\limits_{j = 1}^N {{w_{ij}}} {f_j}(s) - {\lambda _i}G(s)K\sum\limits_{j = 1}^N {{w_{ij}}} {e_j}(s)]. \label{eq:426}
\end{align}

Using $x(s) = (T \otimes {{\cal I}_n})\hat x(s)$, one has
\begin{align} \label{eq:427}
{x_i}(s) = \sum\limits_{m = 1}^N {{v_{im}}} {\hat x_m}(s),
\end{align}
and thus, the state of agent $i$ becomes
\begin{align} \label{eq:428}
&{x_i}(s) =
\sum\limits_{m = 1}^N {{v_{im}}} {\left[ {{{\cal I}_n} + {\lambda _m}G(s)K} \right]^{ - 1}}G(s) \times  \nonumber \\
& \qquad \qquad \qquad \qquad [\sum\limits_{j = 1}^N {{w_{mj}}} {f_j}(s) - {\lambda _m}K\sum\limits_{j = 1}^N {{w_{mj}}} {e_j}(s)].
\end{align}


After some algebraic manipulation, one has
\begin{align}\label{eq:431}
&\sum\limits_{k \in N_i^I}^{} {({x_i}(s) - {x_k}(s))}  = \nonumber \\
&\qquad \sum\limits_{k \in N_i^I}^{} {{\Psi _{ik}}(s)}  + \sum\limits_{k \in N_i^I}^{} {{{\bar \Psi }_{ik}}(s){f_i}(s)\,}  - \sum\limits_{k \in N_i^I}^{} {{{\bar {\bar \Psi} }_{ik}}(s)} {e_i}(s),
\end{align}
where
\begin{align} 
&{\Psi _{ik}}(s) = \sum\limits_{m = 2}^N {({v_{im}} - {v_{km}})} {[ {{{\cal I}_n} + {\lambda _m}G(s)K}]^{ - 1}}G(s) \times \nonumber  \\
& \qquad \qquad \qquad \qquad \qquad  [ {\sum\limits_{j \in {N_{ - i}}}^{} {{w_{mj}}} ({f_j}(s) - {\lambda _mK}{e_j}(s))} ],  \label{eq:432*} \\
&{\bar \Psi _{ik}}(s) = \sum\limits_{m = 2}^N {({v_{im}} - {v_{km}})} {\left[ {{{\cal I}_n} + {\lambda _m}G(s)K} \right]^{ - 1}}G(s){w_{mi}}, \label{eq:432**}  \\
&{\bar {\bar \Psi} _{ik}}(s) = \sum\limits_{m = 2}^N {({v_{im}} - {v_{km}})} {\left[ {{{\cal I}_n} + {\lambda _m}G(s)K} \right]^{ - 1}}G(s) {{\lambda _m}{w_{mi}}}. \label{eq:432}
\end{align}
and $ {N_{ - i}} = {\cal V} \backslash \{ {v_i}\}$.

We now show that \eqref{eq:432*}-\eqref{eq:432} are bounded, regardless of the attack signal ${f_j}(s)$. To this aim, we show that ${\left[ {{{\cal I}_n} + {\lambda _m}G(s)K} \right]^{ - 1}}G(s)$, $\forall m = 2,{\rm{ }}.{\rm{ }}.{\rm{ }}.{\rm{ }},N$ is Hurwitz, i.e., the poles of ${\left[ {{{\cal I}_n} + {\lambda _m}G(s)K} \right]^{ - 1}}G(s)$ are equal to the roots of the characteristic polynomial $A - {\lambda _m}BK$. After some manipulation, one has
\begin{align} \label{eq:430}
\begin{array}{l}
det{\rm{ }}(s{{\cal I}_n} - (A - {\lambda _m}BK)) = det(s{{\cal I}_n} - A + {\lambda _m}BK)\\
\,\,\,\,\,\,\,\,\,\,\,\,\,\,\,\,\,\,\,\,\,\,\,\,\,\,\,\,\,\,\,\,\,\, = det(s{{\cal I}_n} - A)det({{\cal I}_n} + {\lambda _m}{(s{{\cal I}_n} - A)^{ - 1}}BK)\\
\,\,\,\,\,\,\,\,\,\,\,\,\,\,\,\,\,\,\,\,\,\,\,\,\,\,\,\,\,\,\,\,\,\, = det(s{{\cal I}_n} - A)det({{\cal I}_n} + {\lambda _m}G(s)K),
\end{array}
\end{align}
which shows that the eigenvalues of $A - {\lambda _m}BK$ are equal to the poles of ${\left[ {{{\cal I}_n} + {\lambda _m}G(s)K} \right]^{ - 1}}G(s)$.  Therefore, ${\left[ {{{\cal I}_n} + {\lambda _m}G(s)K} \right]^{ - 1}}G(s)$ is Hurwitz and consequently, \eqref{eq:432*}-\eqref{eq:432} are bounded.
 Now, utilizing \eqref{eq:5} and \eqref{eq:431} and the fact that \eqref{eq:432*}-\eqref{eq:432} are bounded, one can observe that if 
\begin{align} \label{eq:433}
{f_i}(s) > {( {\sum\limits_{k \in N_i^I}^{} {{{\bar \Psi }_{ik}}(s)} } )^{ - 1}}( {( {\frac{{\cal I}_m}{{\sqrt {{\eta _i}} }} + \sum\limits_{k \in N_i^I}^{} {{{\bar {\bar \Psi} }_{ik}}(s)} } ){e_i}(s) - \sum\limits_{k \in N_i^I}^{} {{\Psi _{ik}}(s)} } ),
\end{align}
then
\begin{align} \label{eq:434}
\sum\limits_{k \in N_i^I}^{} {{\Psi _{ik}}(s)}  + \sum\limits_{k \in N_i^I}^{} {{{\bar \Psi }_{ik}}(s){f_i}(s)\,}  - \sum\limits_{k \in N_i^I}^{} {{{\bar {\bar \Psi} }_{ik}}(s)} {e_i}(s) > \frac{{\cal I}_m}{{\sqrt {{\eta _i}} }}{e_i}(s) ,
\end{align}
and consequently
\begin{align} \label{eq:435}
 {\eta _i}(\sum\limits_{k \in N_i^I}^{} {({x_i}(s) - {x_k}(s)})^2) - {({e_i}(s))^2} > 0.
\end{align}
This causes the network to exhibit continuous-triggering misbehavior. Therefore, by injecting a constant attack signal with sufficiently large magnitude into the agent $i$, an attacker can cause the triggering condition \eqref{eq:5} to be permanently violated, resulting in the undesirable phenomenon of Zeno behavior. This completes the proof.
\end{proof}

\section{SIMULATION RESULTS}
In this section, two examples are provided to illustrate theoretical results of the previous section. 

\subsection{Strategic sensor attack on CS-ETM}
In this subsection, the effects of the strategic malicious attack on sensor on the CS-ETM is analyzed. Assume a group of 8 agents with single integrator dynamic, i.e., $A=0$ and $B=1$ in (\ref{eq:1}), communicating with an undirected graph topology depicted in Fig. \ref{fig:Fig6}.

\begin{figure}[!ht]
\begin{center}
\includegraphics[width=2in,height=0.7in]{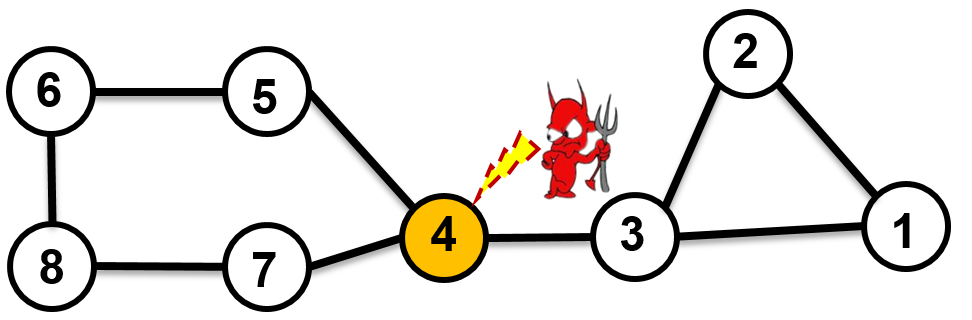}
\vspace{-5pt}
\caption{The Communication Graph.}\label{fig:Fig6}
\captionsetup{justification=centering}
\vspace{-5pt}
\end{center}
\end{figure}

The control protocol \eqref{eq:9} and the measurement error \eqref{eq:7} is used. The triggering condition \eqref{eq:8} is used with ${\eta _1} = {\eta _2} = {\eta _3} = {\eta _4}= {\eta _5}= {\eta _6}= {\eta _7}= {\eta _8}= 0.01$. The initial condition of agents is assumed to be ${x_1}(0) = 6,{x_2}(0) = 1,{x_3}(0) =  -3,{x_4}(0) =  1,{x_5}(0) =  2,{x_6}(0) =  1,{x_7}(0) =  - 2,{x_8}(0) =  - 5$. It is assumed that Agent 4 is under a strategic replay attack \eqref{eq:300} for ${\rm{t }} \ge {\rm{ 5}}{\rm{.1}}$ Sec. Fig. \ref{fig:Fig7} shows the state of agents. Agent $4$ exhibits no-triggering misbehavior and causing the original graph to cluster into 2 subgraphs, therefore harming the communication graph connectivity. These results illustrate the results of Theorem 1.   
\vspace{-8pt}

\begin{figure}[!ht]
\begin{center}
\includegraphics[width=3.3in,height=1.8in]{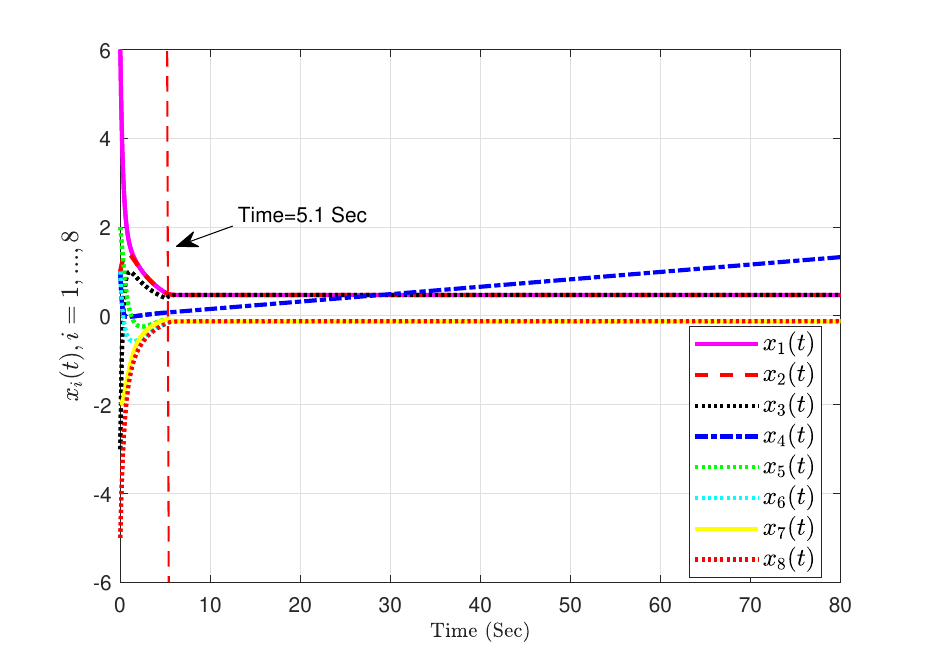}
\caption{The state of agents with the CS-ETM control protocol when Agent 4 is under a strategic replay attack \eqref{eq:300} on its sensors.}\label{fig:Fig7}
\captionsetup{justification=centering}
\end{center}
\end{figure}
\vspace{-15pt}

\subsection{Strategic actuator attack on S-ETM}
In this subsection, the effects of a strategic malicious attack on the S-ETM is illustrated. Assume that there are 4 agents with single integrator dynamics, i.e., $A=0$ and $B=1$ in (\ref{eq:1}), communicating over the graph topology depicted in Fig. \ref{fig:Fig3}.
\vspace{-8pt}

\begin{figure}[!ht]
\begin{center}
\includegraphics[width=1.45in,height=1in]{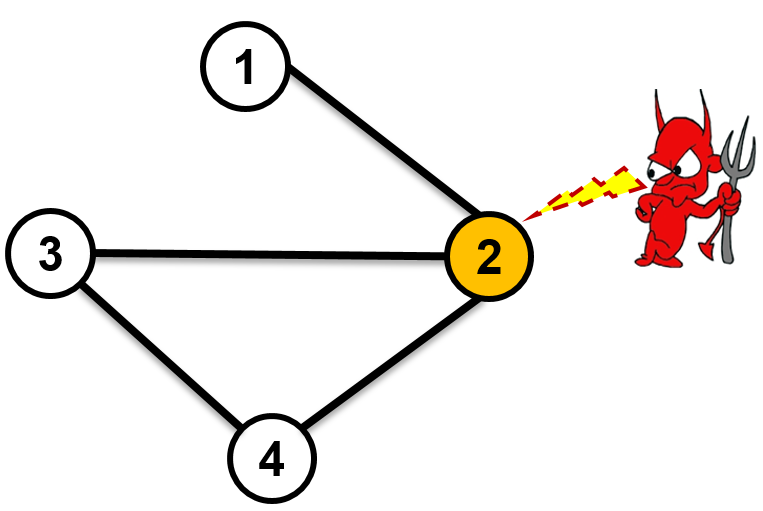}
\caption{The Communication Graph.}\label{fig:Fig3}
\captionsetup{justification=centering}
\vspace{0pt}
\end{center}
\end{figure}

\begin{figure}[!ht]
\begin{center}
\includegraphics[width=3.3in,height=1.8in]{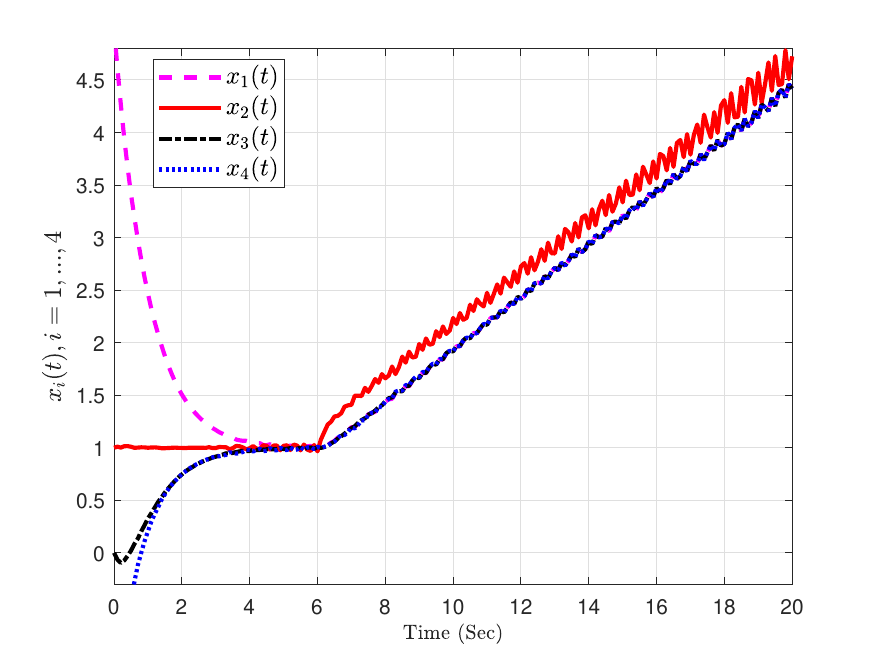}
\vspace{0pt}
\caption{The state of agents with the S-ETM control protocol when Agent 2 is under a strategic malicious attack on its actuator.}\label{fig:Fig4}
\captionsetup{justification=centering}
\vspace{0pt}
\end{center}
\end{figure}
\vspace{-15pt}

The control protocol \eqref{eq:6} and the measurement error \eqref{eq:4} is used. The triggering condition \eqref{eq:5} is used with ${\eta _1} = {\eta _2} = {\eta _3} = {\eta _4} = 0.01$. The initial conditions of agents are assumed to be ${x_1}(0) = 5,{x_2}(0) = 1,{x_3}(0) = 0,{x_4}(0) =  - 2$. Now, let Agent 2 be under a constant actuator attack with the signal $f_2(t)=-1$ for $t > 6$ Sec.  Fig. \ref{fig:Fig4} and Fig. \ref{fig:Fig5} show the state and the measurement error \eqref{eq:4} for all agents. It can be seen that before the attacks all agents reach consensus and the measurement error converges to zero. This implies that the entire network reached the desired consensus value and no further triggering event is required. However, when a strategic malicious attack on the actuator of Agent 2 is launched, all agents start to diverge and the network shows a continuous-triggering misbehavior, as shown in Fig. \ref{fig:Fig5}. These results illustrate the results of Theorem \ref{Theor2}.   

\begin{figure}[!ht]
\begin{center}
\includegraphics[width=3.3in,height=2in]{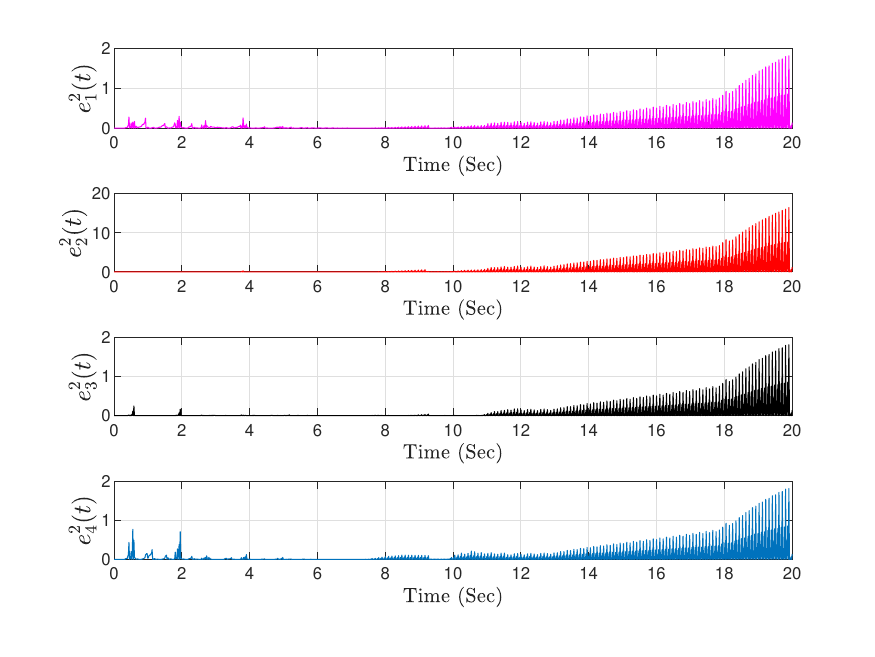}
\vspace{0pt}
\caption{The square of measurement error \eqref{eq:4} in the S-ETM. The entire network shows continuous triggering misbehavior.}\label{fig:Fig5}
\captionsetup{justification=centering}
\vspace{0pt}
\end{center}
\end{figure}
\vspace{-15pt}

\section{CONCLUSION}
The adverse effects of cyber-physical attacks on the performance of distributed MASs with the event-triggered controller are analyzed. Both combinational state-based event mechanism (CS-ETM) and state-based event mechanism (S-ETM) are considered and the effect of attacks on both event mechanisms are studied. It is shown that an attacker can design a strategic malicious attack on actuators and sensors to falsify both event mechanisms. In the CS-ETM, it affects the triggering condition in the sense that no events are triggered, while the team of agents does not reach consensus. Furthermore, in the S-ETM, it fools the event-triggered mechanism to continuously generate triggering events, and thus, resulting in the undesirable phenomenon of Zeno behavior.

\bibliographystyle{IEEEtran}
\bibliography{IEEEabrv,ref.bib}

\vfill

\end{document}